\documentclass[a4paper]{article}

\usepackage{inputenc, amsthm, amsmath} 
\usepackage{caption, subcaption, amssymb}
\usepackage{algpseudocode, algorithm}
\usepackage{calc, enumerate, wasysym}
\usepackage{graphicx, url, textcomp}
\usepackage{epstopdf}
\usepackage{array}

\newlength{\mywidth}
\newcommand\bigfrown[2][\textstyle]
{\ensuremath{%
		\addtolength{\arraycolsep}{-2mm}
		\array[b]{c}\text{\resizebox{\mywidth}
			{.8ex}{$#1\frown$}}\\[-1.6ex]#1#2\endarray}}
\newtheorem{theorem}{Theorem}[section]
\theoremstyle{definition}

\newtheorem{lemma}[theorem]{Lemma}

\newtheorem{remark}[theorem]{Remark}

\newcommand*{\email}[1]{\texttt{#1}}
\providecommand{\keywords}[1]{\textbf{\textit{Keywords:}} #1}

\begin{document}
\title{Extending editing capabilities
	of subdivision schemes by refinement of point-normal pairs}
\author{
	Evgeny Lipovetsky\footnote{Corresponding author}
	\footnote{The contribution of Evgeny Lipovetsky is part of his
		Ph.D.  
		research conducted at Tel-Aviv University.}
	$^{,}$\thanks{\email{evgenyl@post.tau.ac.il}, 
		School of Computer Sciences, Tel-Aviv Univ.,Israel} \and 
	Nira Dyn\thanks{\email{niradyn@post.tau.ac.il}, School of Mathematical Sciences, Tel-Aviv Univ., Israel}}	
\maketitle	
	\begin{abstract}
In this paper we extend the 2D circle average of \cite{ld:16} 
to a 3D binary average of
point-normal pairs, and study its properties. We modify classical 
surface-gener- ating linear subdivision schemes with this average
obtaining surface-generating schemes refining point-normal pairs. 
The modified schemes give the possibility to generate more geometries
by editing the initial normals.
For the case of input data consisting of a mesh only, we present a method
for computing "naive" initial normals from the initial mesh.
The performance of several modified schemes is compared to their linear
variants, when operating on the same initial mesh,
and examples of the editing capabilities of the modified schemes
are given.
In addition we provide a link
to our repository, where we store the initial and refined mesh files, and the implementation code. Several videos, demonstrating the editing capabilities of the initial normals
are provided in our Youtube channel.
	\end{abstract}
	
\keywords
{		
	surface-generating 
	subdivision refining 3D point-normal pairs,
	3D circle average,
	surface design.
}

\section{Introduction}
In a previous paper \cite{ld:16} we introduced a weighted binary average 
of two 
2D point-normal pairs (PNPs),
termed {\em circle average}, and defined subdivision schemes based on it, which 
refine PNPs in 2D, and generate curves. 
This paper presents a method for extending the 2D circle average to 3D,  
a method applicable to any 2D weighted binary average.
The extension of the 2D circle average to 3D preserves several important 
properties of the 2D circle average, and even extends some of them.

With the 3D circle average we modify classical linear 
surface-generating
subdivision schemes refining points, to surface-generating schemes refining 
PNPs.
Our methodology in modifying linear schemes, is the same as that in the 2D 
case:
we first write the refinement rules of a converging linear scheme in terms of 
repeated weighted linear binary averages, and then replace each weighted linear binary average by the 3D circle average with the same weight. 
As in the case of the modified 2D schemes \cite{ld:16},  
our new 3D schemes also enrich the variety of geometries that 
can be generated, 
just by editing the initial normals.                                                                                
\subsection{Contribution}
\label{subsec:contribution}
With subdivision schemes refining points, the only way to alter
the geometry generated from a given initial mesh is to change the location
of the vertices of the mesh. The schemes we design in this work
refine point-normal pairs. Such a scheme can generate a richer variety of geometries by editing also the initial normals.

With our approach, we can modify any convergent linear scheme refining
points into a scheme refining PNPs. We can enrich
any subdivision-based software system by establishing new editing tools.
\subsection{Outline}
\label{sub:outline}
Sections \ref{sec:ca_def}, \ref{sec:prop_3D_ca} present the extension 
to 3D of the 2D 
circle average, and 
asserts properties of the 3D circle average, in particular the Consistency 
Property, which makes the weighted binary operation an average. 
In Section \ref{sec:subdivision} we explain our methodology for modifying linear subdivision 
schemes refining points to schemes refining PNPs. We give a method  
(algorithm) to write a weighted linear average of several elements in terms 
of repeated 
weighted linear binary averages. We also present a method for defining "naive" 
normals from a mesh  when initial normals are not given.
In Section \ref{sec:performance} four modified surface-generating schemes are discussed. 
Their performance on two initial meshes is compared
with that of their linear counterparts. 
The editing capability of the modified schemes is demonstrated by examples
presented in three videos and a figure of three snapshots from one of the videos.
A link to the videos is given.
A short discussion of the implementation consists of Section \ref{sec:implem}.
The paper ends in Section \ref{sec:future} with Conclusions and Future Work.
\section{Related work}
Classical surface-generating linear subdivision schemes refine points given 
at the vertices of a mesh \cite{pr:08}. In recent years more involved data at the vertices 
of a mesh are refined by subdivision schemes, such as by Hermite schemes 
(see e.g. \cite{ms:12}), and by manifold-valued schemes (see e.g.
 \cite{w:14}). 
The information in 3D PNPs is less than that in 3D data refined by Hermite 
schemes. In the latter case the data is "linear" and is refined by a linear scheme, while  the unit normals of 
a surface is a nonlinear functional applied to the surface (it is the 
direction of the cross-product of two tangents). Indeed the circle average
and the modified schemes based on it are  nonlinear.

Previous works on subdivision schemes that refine PNPs are based on a 
binary operation between two 
PNPs which is used to define the insertion rule in an interpolatiory scheme
\cite{jue:07p}, \cite{aihjjg:16}.

There are various papers using the same methodology as ours in adapting
linear refining-points subdivision schemes to schemes refining other types of geometric objects.
For example: manifold-valued data (see e.g. \cite{ds:14}),
sets in $\mathbb{R}^d$ (see e.g. \cite{kd:13}), and
nets of functions (see e.g. \cite{cd:11}).
\section{The circle average}
\label{sec:ca_def}
In this section we introduce an extension of the 2D \textit{circle average} 
- a weighted binary average of two 2D point-normal pairs (PNPs) - to an average
of two 3D PNPs. First we recall the definition of the 2D circle average.

\subsection{The circle average in 2D}
Given two PNPs in 2D, each consisting of a point and a normal 
unit vector, $P_0 = (p_0,n_0), P_1 = (p_1, n_1)$, and a real weight 
$\omega \in [0,1]$, the circle average produces a new PNP $P_\omega = 
P_0\circledcirc_\omega P_1 = (p_\omega, n_\omega)$.

The point $p_\omega$ is on an auxiliary arc $\bigfrown{P_0 P_1}$,
at arc distance $\omega\theta$ from $p_0$,
where $\theta$ is the angle between $n_0$ and $n_1$.
The normal
$n_\omega$ is the geodesic average of $n_0$ and $n_1$.
For the definition of $\bigfrown{P_0 P_1}$ and for
more details consult \cite{ld:16}.
\subsection{Construction of the circle average in 3D}
All the objects mentioned in the rest of the paper, specifically points and 
vectors, are in 3D, if not stated otherwise.

First, we introduce some notation.
For two vectors $u,v$, with $u \times v \ne 0$,
$z(u, v)$ denotes the normalized vector 
in direction $u \times v$.
Note that 
\begin{align}
z(\alpha u + \beta v, \gamma u + \delta v) = z(u,v), 
\text{ for }
\alpha, \beta, \gamma, \delta \in \mathbb{R},
\text{ s.t. } \alpha^2 + \beta^2 > 0, 
\gamma^2 + \delta^2 > 0.
\label{eq:cross_preserv}
\end{align}
For a point $p$ and a vector $n$, let $\Pi(p,n)$ be the plane 
which passes through the point $p$ and has the normal $n$. 

For $P_0 = (p_0, n_0)$ and $P_1 = (p_1, n_1)$,
two PNPs to be averaged, we consider the two parallel planes 
$\Pi_0 = \Pi(p_0, z(n_0, n_1))$, $\Pi_1 = \Pi(p_1, z(n_0, n_1))$.
The length of the projection of $[p_0, p_1]$ on 
$z(n_0, n_1)$, which is the distance between $\Pi_0$ and $\Pi_1$,
is denoted by $\hbar$ . 
We define $\Pi_\omega$ to be the plane 
parallel to $\Pi_0$, $\Pi_1$, which is at distance 
$\omega\hbar$ from $\Pi_0$ towards $\Pi_1$.
We say that $P = (p, n)$ belongs to a plane $\Pi, P \in \Pi$, if
both $p$ and $n$ are in $\Pi$.

Let $P_\omega = P_0 \circledast_\omega P_1$ denote the circle average in 3D.
The construction of $P_\omega = (p_\omega, n_\omega)$ is done by 
the following procedure. 
(See Fig.~\ref{fig:construction3D} for an example.)
\begin{figure}[h]
	\centering
	\includegraphics[trim={0cm 0cm 0cm 1cm},clip,scale=1.1]
	{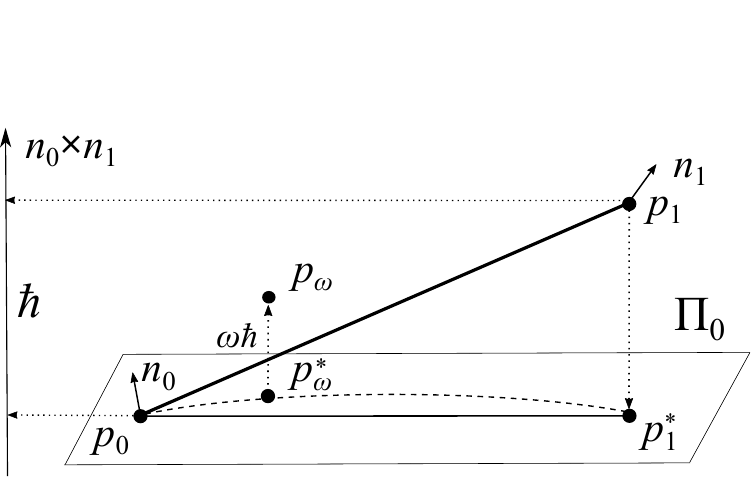} 
	\captionsetup{justification=centering}
	\caption{Construction of $P_0 \circledast_{\omega} P_1$ in 3D.}
	\label{fig:construction3D}
\end{figure}
\begin{enumerate}[(i)]
	\item Project $p_1$ on $\Pi_0$ to obtain $p^*_1$. 
	Note that $P^*_1 = (p^*_1, n_1) \in \Pi_0$,
	and that also $P_0\in\Pi_0$.
	\item Compute the 2D circle average 
	$ P_0 \circledcirc_\omega P^*_1$
	in a local coordinate system in $\Pi_0$, and
	convert the 2D result to a PNP 
	$P^*_\omega = (p^*_\omega, n_\omega)$ in 3D.
	\item Project $p^*_\omega$ on $\Pi_\omega$.
	Obtain $P_\omega = (p_\omega, n_\omega)$.
\end{enumerate}
Note that if $P_0, P_1 \in \Pi_0$ then the 3D average 
reduces to the 2D average.
As in the 2D case, the construction is not defined when $\theta = \pi$.
Although, when $\theta = 0$, the direction $z$ is not defined, the 3D average
can be obtained by continuity as shown in Section \ref{subsec:circ_as_lin}.
\begin{remark}
	It is easy to see that any 2D average of two PNPs  
	can be extended by the above method to a 3D average.
\end{remark}

\section{Properties of the 3D circle average}
\label{sec:prop_3D_ca}
Many of the properties of the 2D circle average are preserved 
or enhanced by the 3D circle average.

\subsection{Consistency}
In \cite{ld:16}, we prove the \textit{\textbf{consistency}} property 
of the circle average in 2D.
Here we argue that the 3D circle average also has this property, namely
$\forall t,s,k \in [0,1]$,
\begin{align}
(P_0 \circledast_t P_1) \circledast_k (P_0 \circledast_s P_1) 
= P_0 \circledast_{\omega^*} P_1,\  \omega^* = ks+(1-k)t.
\label{eq:consistency}
\end{align}
Indeed, the consistency of the normals
is guaranteed by the consistency of the geodesic average. 
For the consistency of the points observe that
by (\ref{eq:cross_preserv}) the projection direction of steps (i) and (iii)
in the construction 
of the 3D circle average is the same in all the averaging operations
in (\ref{eq:consistency}). 
Also, all averages in (\ref{eq:consistency}), 
except for $\circledast_k$,
are performed in $\Pi_0$, and then moved in the
direction $z(n_0, n_1)$, by a length which is a linear average between
zero and $\hbar$. 
The average $\circledast_k$ is performed in a plane parallel to $\Pi_0$
translated in the $z(n_0, n_1)$ direction by an appropriate linear average 
of zero and $\hbar$.
Since both the 2D circle average and the linear average
have the consistency property, and they are performed independently 
in orthogonal directions,
the consistency of the points is asserted. 
See Figure \ref{fig:consistency} for an example.
\begin{figure*} 
	\qquad
	\begin{subfigure}[b]{0.4\textwidth}
		\centering
		\includegraphics[scale=1.1]
		{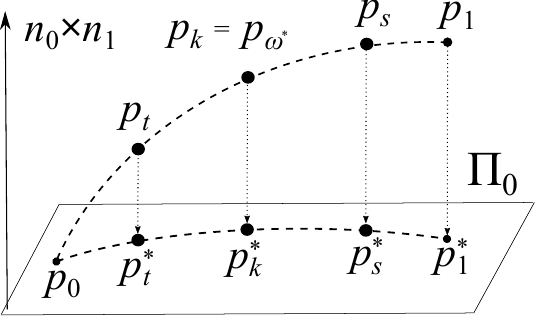} 
		\caption{Consistency of points.}
	\end{subfigure}
	\qquad\qquad
	\begin{subfigure}[b]{0.4\textwidth}
		\centering
		\includegraphics[scale=1.1]
		{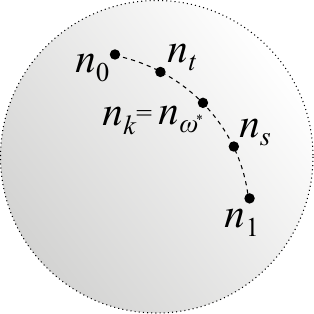} 
		\caption{Consistency of normals.}
	\end{subfigure}
	\captionsetup{justification=centering}
	\caption{Consistency property. \\
		Note that $z$ and $\bigfrown{P_0P_1^*}$ are
		the same for the circle average with weights $t,s,k$.}
	\label{fig:consistency}
\end{figure*}

\subsection{Helix trace}
If we change the weight in $\circledast_\omega$ 
continuously and track
the position of the points of $P_0 \circledast_\omega P_1$, 
then, in a generic 3D case,
we obtain a helix instead of the arc 
$\bigfrown{P_0 P_1}$
in the 2D case. We denote this helix by $H(P_0, P_1)$. Note that
the projection on $\Pi_0$ of $H(P_0, P_1)$ is $\bigfrown{P_0P_1^*}$ (see Figure \ref{fig:consistency}a).

\subsection{Limit cases of the circle average}
\label{subsec:circ_as_lin}
The investigation of several limit cases 
of the circle average requires the next lemma. The proof of this lemma
is straightforward but we give it for the convenience of the reader.
\begin{lemma}
	The intersection point $x_{\omega}$ of $[p_0, p_1]$ and
	$\Pi_\omega$ is given by
	\begin{align}
	x_\omega = (1 - \omega) p_0 + \omega p_1.
	\label{eq:proj_inx_work_plane}
	\end{align}
	\label{lemma:lin_lemma}
\end{lemma}
\begin{proof}
	Figure \ref{fig:lin_lemma} illustrates the  proof.
	Let $\widehat{p_0}, \widehat{p_1}$ be the projections of $p_0,p_1$ on 
	$\Pi_\omega$, and let $\widehat{\Pi} = \Pi(p_0, \widehat{n})$ where $\widehat{n}$ is the normalized
	vector $\overrightarrow{p_0p_1} \times 
	\overrightarrow{\widehat{p_0}\widehat{p_1}}$. Note that $\widehat{\Pi}$
	contains the segments $[p_0, p_1]$ and $[\widehat{p_0}, \widehat{p_1}]$.
	Since $[\widehat{p_0}, \widehat{p_1}]$ is in $\Pi_\omega$ then
	$[\widehat{p_0}, \widehat{p_1}]$ is contained in the intersection
	line of $\widehat{\Pi}$ and $\Pi_\omega$.
	Moreover, by the definition of $\widehat{p_0}, \widehat{p_1}$,
	$[\widehat{p_0}, \widehat{p_1}]$ contains
	all the projections of points of $[p_0, p_1]$ on $\Pi_\omega$,
	in particular
	$x_\omega \in [\widehat{p_0}, \widehat{p_1}]$. Thus the two
	triangles
	$\triangle p_0 \widehat {p_0}x_\omega$ and $\triangle p_1 \widehat{p_1}
	x_\omega$ are in $\widehat{\Pi}$. 
	These triangles are similar, having equal angles. 
	Therefor  
	$\frac{|p_0x_\omega|}{|p_1x_\omega|} = \frac{|p_0\widehat{p_0}|}{|p_1\widehat{p_1}|} = 
	\frac{\omega}{1-\omega}$,
	which proves (\ref{eq:proj_inx_work_plane}).
\end{proof}
\begin{figure}[h]
	\centering
	\includegraphics[trim={0cm 0.5cm 0cm 0cm},clip,scale=1.1]
	{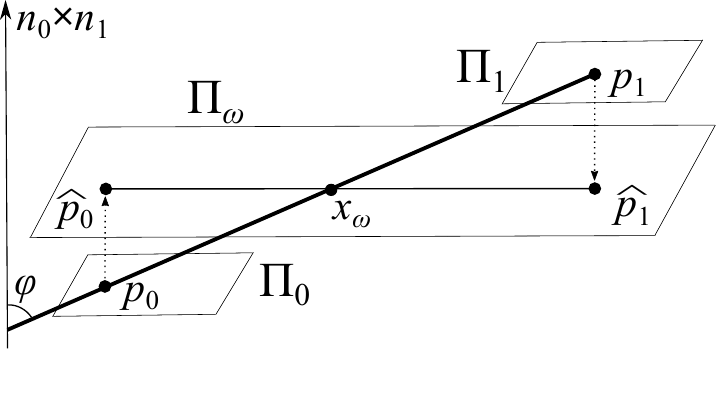} 
	\captionsetup{justification=centering}
	\caption{The setup of Lemma \ref{lemma:lin_lemma}.
		The intersection point of $[p_0, p_1]$ and $\Pi_\omega$
		is the linear average of $p_0$ and $p_1$ with weight $\omega$.}
	\label{fig:lin_lemma}
\end{figure}
Two basic properties of the 2D circle average are that 
it is not defined for $\theta = \pi$, and that
its point tends to
the 2D linear average
when $\theta \rightarrow 0$. A similar behavior holds in 3D.
Furthermore, in the 3D case, there are two parameters,
$\theta$ and the angle $\varphi$ between $n_0 \times n_1$ and
$\overrightarrow{p_0p_1}$ (see Figure \ref{fig:lin_lemma}).
Note that $\theta \in [0, \pi)$ and $\varphi \in [0, \pi]$.

Next we show that the point of $P_0 \circledcirc_{\omega} P_1$
tends to $x_\omega = (1-\omega)p_0 + \omega p_1$ as either 
$\theta \rightarrow 0$ or $\varphi \rightarrow 0$ (or $\varphi \rightarrow \pi$).

First we analyze the case $\varphi \rightarrow 0$ 
or $\varphi \rightarrow \pi$
for fixed $\theta \ge 0$.
\begin{figure}[h]
	\centering
	\includegraphics[scale=1]
	{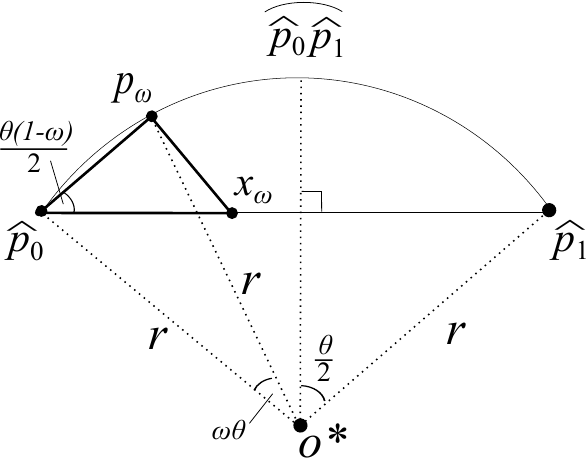} 
	\captionsetup{justification=centering}
	\caption{The triangle $\triangle\widehat{p_0} p_\omega x_\omega$.}
	\label{fig:cos_thrm}
\end{figure}

Using the geometry 
as depicted in Figure \ref{fig:cos_thrm},
we get $|\widehat{p_0}p_\omega| = |\widehat{p_0}\widehat{p_1}|
\dfrac{\sin\left(\frac{\theta\omega}{2}\right)}
{\sin\left(\frac{\theta}{2}\right)}$.
We express 
$|p_\omega x_\omega|$ by the cosine theorem in the triangle 
$\triangle\widehat{p_0} p_\omega x_\omega \subset \Pi_\omega$,
\begin{align}
|p_\omega x_\omega|^2 &=
(|\widehat{p_0}\widehat{p_1}|\omega)^2 +
\bigg(|\widehat{p_0}\widehat{p_1}|
\dfrac{\sin\left(\frac{\theta\omega}{2}\right)}
{\sin\left(\frac{\theta}{2}\right)}\bigg)^2 \notag \\
&-2|\widehat{p_0}\widehat{p_1}|^2\omega
\dfrac{\sin\left(\frac{\theta\omega}{2}\right)}
{\sin\left(\frac{\theta}{2}\right)}
\cos\left(\frac{\theta\left(1-\omega\right)}{2}\right).
\label{eq:cos_in_work_plane}
\end{align}
Since $|\widehat{p_0}\widehat{p_1}| = |p_0 p_1| \sin \varphi$,
all the terms in the right side of
(\ref{eq:cos_in_work_plane}) tend to zero, independently of the value of 
$\theta$.
Thus, $p_\omega \rightarrow x_\omega$, as
$\varphi \rightarrow 0$ (or $\varphi \rightarrow \pi$).

In case $\theta \rightarrow 0$ and $0 < \varphi < \pi$, the right side of
(\ref{eq:cos_in_work_plane}) is zero because 
\begin{align}
\lim_{\theta \rightarrow 0}{\dfrac{\sin\left(\frac{\theta\omega}{2}\right)}
	{\sin\left(\frac{\theta}{2}\right)}} = \omega.
\end{align}

\subsection{Preservation of special geometries}
In \cite{ld:16} it is shown that if $P_0, P_1$ are sampled from
a circle then $P_0 \circledcirc_\omega P_1 = (p_\omega, n_\omega)$
corresponds to a point on this circle with $n_\omega$ the normal
of the circle at $p_\omega$. We say that the 2D circle average
"preserves circles". This property extends in the case of the 3D  circle average to
"preservation of spheres and cylinders". 
Indeed, any two PNPs sampled from
a sphere are also samples from the big circle $C$, determined by the two
points and the center of the sphere. Thus, the circle average of the two PNPs
is the 2D circle average of the two PNPs sampled from $C$, implying that the
3D circle average "preserves" spheres. 
See Figure \ref{fig:geom_preserv}a for 
an example. 

Next, consider the case that the two PNPs $P_0, P_1$ are sampled 
from a cylinder
of the form $x = \cos t, y = \sin t, z = t, 0 \leq t \le \pi$.
Note that $z(n_0, n_1)$ is the axis of the cylinder, and that
$H(P_0, P_1)$ is on the cylinder, and its projection on $\Pi_0$
is $\bigfrown{P_0P_1^*}$.
Due to the "circle preservation" of the 2D circle average,
$P_0 \circledcirc_\omega P_1^* = (p^*_\omega, n_\omega)$ corresponds
to a PNP sampled from $\bigfrown{P_0P_1^*}$. 
By (iii) of the construction of the 3D circle average,
$p_\omega$ is on $H(P_0, P_1)$ and the normal to $H(P_0,P_1)$ 
at this point is $n_\omega$.
Thus the 3D circle average "preserves" cylinders. 
See Figure \ref{fig:geom_preserv}b for an example. 
\begin{figure*} [!htb]
	\qquad
	\begin{subfigure}[b]{0.4\textwidth}
		\centering
		\includegraphics[scale=1.1]
		{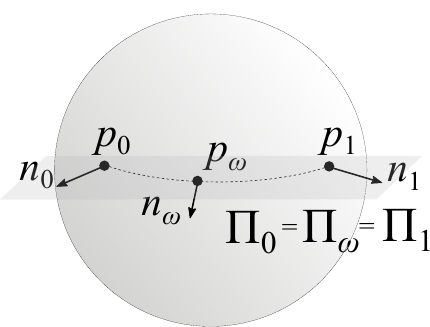} 
		\caption{Sphere preservation.}
	\end{subfigure}
	\qquad\qquad
	\begin{subfigure}[b]{0.4\textwidth}
		\centering
		\includegraphics[scale=1.1]
		{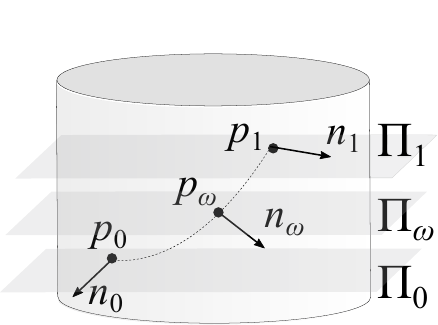} 
		\caption{Cylinder preservation.}
	\end{subfigure}
	\captionsetup{justification=centering}
	\caption{Preservation of special geometries.}
	\label{fig:geom_preserv}
\end{figure*}
\section{Modified subdivision schemes}
\label{sec:subdivision}
We consider subdivision schemes refining point-normal
pairs, which are obtained from converging linear subdivision schemes.
To obtain these schemes we express the linear schemes in terms of repeated,
weighted, linear, binary averages of points and replace these averages by the 
3D circle average.
In this paper we term the so obtained schemes "Modified schemes".

We first explain how we rewrite the refinement rules of any linear 
converging subdivision scheme in terms 
of repeated, weighted, linear, binary averages.
Then we mention four classical surface-generating linear schemes to 
be modified. 
The performance of these schemes and their 
modifications is tested in the next section.
Finally we provide a method 
which computes initial normals, if normals are not given as input.

\subsection{Repeated binary averaging}
\label{subsec:repeated_bin_avg}
Here we propose a method for rewriting a weighted linear average
of several points in terms of 
repeated binary averages.
Consider computing a point $q$ as a weighted average of points
$\{p_i\}_{i=0}^{k-1}$, i.e.
\begin{align}
q = \alpha_{0}p_{0} + \alpha_{1}p_{1} + ... + \alpha_{k-1}p_{k-1},
\label{eq:subd_rule}
\end{align}
where $\alpha_i \in \mathbb{R}, \alpha_i \ne 0, i = 0,...,k-1,$ and
\begin{align}
\sum_{i = 0}^{k-1}{\alpha_i} = 1.
\label{eq:sum_coefs_1}
\end{align}
We rewrite (\ref{eq:subd_rule}) as
\begin{align}
q = (\alpha_{0} + \alpha_{1})
\Big( \frac{\alpha_{0}}{\alpha_{0} + \alpha_{1}}p_{0} + 
\frac{\alpha_{1}}{\alpha_{0} + \alpha_{1}}p_{1}\Big) + 
\alpha_{2}p_{2} + 
... + \alpha_{k-1}p_{k-1},
\label{eq:repeated_bin_av_at_1st}
\end{align}
reducing by one the number of elements in the outer sum and obtaining a 
binary average as the first term. Note that in
(\ref{eq:repeated_bin_av_at_1st}) q is a linear average of $k-1$ points
while in (\ref{eq:subd_rule}) q is a linear average of $k$ points.
We repeat this step
$k-2$ times and obtain the weighted average (\ref{eq:subd_rule}) written as 
a sequence of $k-1$ repeated binary averages. 

To avoid division by zero in this process, we have to guarantee that
\[\sum_{i = 0}^{\ell}{\alpha_i} \neq 0, \ \ell = 1, ..., k-1.\]
We reorder
the terms in (\ref{eq:subd_rule}) such that all positive $\alpha_i$
precede all the negative ones. With this reordering, (\ref{eq:sum_coefs_1})
guarantees that each partial sum 
$\sum_{i = 0}^{\ell}{\alpha_i}$ is positive.

It seems that a challenge is to find an order of summands in (\ref{eq:subd_rule})
that performs best.
Yet our experiments indicate that the performance of a modified scheme is
almost independent of the order of the summands.

\subsection{The linear schemes to be modified}
\label{subsec:mod_subd}
In this paper we modify four classical surface-generating linear subdivision schemes:
Catmull-Clark (CC) \cite{cc:1978}, Kobbelt 4-point (K4) \cite{kob_4pt:1996}, Butterfly (BY) \cite{dgl:90}, and Loop (LP) \cite{loop:87}. 
The modification of all these schemes is done by the method of Section
\ref{subsec:repeated_bin_avg}. We denote a modified scheme by adding
"M" before the acronym of the corresponding linear scheme.  

\subsection{Naive normals}
\label{subsec:naive_norms}
Here we present a method for determining
initial normals at the vertices of a given control mesh,
when the normals are not given as input.

Consider neighboring faces $\{f_i\}_{i=0}^{k-1}$ and neighboring 
vertices $\{v_i\}_{i=0}^{k-1}$ of 
a given vertex $p$ in 
a control mesh (see Figure \ref{fig:naive_norms}).
Denote by $a_i$ the normalized vector $\overrightarrow{pv_i} \times
\overrightarrow{pv_{i+1}}$.
The unit vector $a_i$ defines a normal related to $f_i$.
Let $\gamma_i$ be the angle $\varangle v_ipv_{i+1}$, and let
$\gamma = \sum\limits_{i=0}^{k-1}{\gamma_i}$. 
We suggest the normalized vector 
\begin{align}
n = \frac{a}{\parallel a \parallel} \quad \text{ with }
a = \sum\limits_{i=0}^{k-1}{\frac{\gamma_i}{\gamma}a_i},
\end{align}
as the "naive normal" at the vertex $p$. This method is similar to the one
discussed in Section 3.5 in \cite{mdsb:02}.
\begin{figure} [!h]
	\centering
	\includegraphics[scale=1.0]{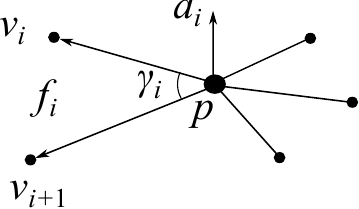} 
	\captionsetup{justification=centering}
	\caption{Determining a naive normal at $p$ from its neighborhood in the mesh.}
	\label{fig:naive_norms}
\end{figure}

\section{Performance of the modified schemes}
\label{sec:performance}
In our examples we consider input meshes that consists of either 
quadrilateral faces or triangular faces. 
We accept meshes with irregular vertices (of valency $\ne 4$ 
in quadrilateral meshes, and $\ne 6$ in triangular meshes).
Our implementation is limited to meshes of exactly one type of faces.
The initial data consists of
the vertices of a mesh with a naive normal attached to each vertex.

\subsection{Comparison methodology}
\label{sec:methodology}
We apply the linear schemes of Section \ref{subsec:mod_subd}
and their modifications on two example meshes with naive normals. 
We compare the
performance of a linear scheme with its modified scheme by inspecting
the generated surfaces and by measuring
estimates of 
dihedral angles and curvatures.
First we explain how we compute these estimates 
using the notation of 
Section \ref{subsec:naive_norms}.
These estimates indicate deviation from $C^1$ smoothness, when considering
the magnitude of the dihedral angles and from $C^2$ smoothness when considering 
the magnitude of the local changes in the 
curvatures.

\subsubsection{Estimating dihedral angles}
In the case of a triangular mesh
we compute the dihedral angle corresponding to an edge as the
angle between the normals of the neighboring triangles
to the edge.

In the case of a quadrilateral mesh we estimate the dihedral angle
at the midpoint of an edge $e$ in the mesh by the following steps:
\begin{enumerate}
	\item Compute the directions $\ell, r$ connecting the midpoint 
	of $e$ with the midpoints of the opposite edges of $e$ in the
	neighboring faces to the left and to the right of $e$, respectively.
	\item Compute $n_\ell = \ell \times \overrightarrow{e}$ and 
	$n_r = r \times \overrightarrow{e}$, where $\overrightarrow{e}$
	is the edge direction.
	\item Compute the angle between $n_\ell$ and $n_r$.
\end{enumerate}
In the following we refer to this angle as the dihedral angle of the edge $e$.
\subsubsection{Measuring discrete curvatures and its local changes}
For every vertex $p$ in a mesh $M$, we compute its discrete curvature 
$K_p$ by
\begin{align}
K_p = \frac{(2\pi - \sum_{i=0}^{k-1}{\gamma_i})}{\mathcal{A}_p},
\quad \text{ with } \mathcal{A}_p = \frac{1}{6}\sum_{i = 0}^{k-1} |pv_i||pv_{i+1}|\sin \gamma_i.
\end{align}
Here $v_k = v_0$, and $\mathcal{A}_p$ stands for the area of the barecentric cell
around $p$ (See e.g. \cite{mdsb:02}.) In the following we refer to $K_p$ as 
the curvature at $p$.

To estimate the magnitude of the local changes in the curvatures,
we first evaluate the curvature at all vertices of the mesh,
and then estimate the local changes
of the curvatures at $p$ as
\begin{align}
	\zeta_p = \big|\max_{v \in V_p}\{K_{v} \} - 
	               \min_{v \in V_p}\{K_{v}\}\big|,
\end{align}
where $V_p$ consists of $p$ and all the adjacent vertices to $p$ in the mesh.

\subsection{Results}
\label{subsec:results}
In studying the performance of the modified schemes, we consider only the 
generated meshes and ignore the generated normals.
Although we do not have a convergence proof for the modified schemes 
investigated in this 
paper, our tests indicate that the generated meshes converge to  a surface.
The convergence of the normals is guaranteed  by general results about 
convergence of manifold-valued subdivision schemes, based on geodesic averages 
(see e.g. \cite{ds:16}).
We do not display the limit of the normals because, as in the 2D case \cite{ld:16},
they are not the normals of the limit surface. Yet, as demonstrated in Table \ref{tbl:angle_err} of Section \ref{subsec:video}, the closer are the initial
normals to the naive normals of the initial mesh, the closer are the limit
normals to the normals of the limit surface.

Two input meshes are studied in this section demonstrating typical performance
of a modified scheme in case of naive initial normals.
One is referred as "Tower" and 
one is referred as "Tube". The Tower mesh is taken both in its quadrilateral
and triangular form.
For each example several iterations of one linear subdivision scheme
and its modified variant are executed. 
The resulting surfaces are depicted in Figures \ref{fig:lp_tower},
\ref{fig:kob4pt_tower}, \ref{fig:butr_tube}.
The colors
indicate the curvature of the final mesh, if not mentioned otherwise.
The colors yellow to red (cyan to blue) indicate positive (negative) values.
All the examples are provided as mesh files in our online repository,
as explained in Section \ref{sec:implem}.
Our observations regarding these examples are based on these files.
See Table \ref{tbl:num_comp} for typical numerical comparisons.
We compute the maximal dihedral angle, $\psi$, and the maximal
$\zeta_p$, $\zeta^*$, for each final mesh.

\begin{center}
	\begin{tabular}{|c|c|c|c|c|c|c|}
		\hline
		scheme & LP & MLP & CC & MCC & K4 & MK4 \\
		\hline
		$\psi$ & 14.78\textdegree & 11.04\textdegree 
		      & 11.02\textdegree & 8.68\textdegree
		      & 37.11\textdegree & 27.98\textdegree \\
		\hline
		$\zeta^*$ & 0.034 & 0.031 
				  & 0.019 & 0.025
				  & 1.231 & 0.846  \\
		\hline
	\end{tabular}
\captionsetup{justification=centering}
	\captionof{table} {Numerical comparisons for the Tower mesh.}
	\label{tbl:num_comp}
\end{center}
The numerical comparisons indicate that $\zeta^*$ of meshes produced 
by modified $C^2$ schemes (MCC, MLP) are of the same order as $\zeta^*$
of their linear counterparts, while for $C^1$ schemes, $\zeta^*$ of the modified
schemes are significantly smaller.
Also, the decrease rate of the maximal estimated dihedral angle
from one refinement level to its next level  is faster in case of
the modified schemes. 
\begin{remark}
	For $P_0, P_1$ with $\theta = 0$ the circle average of the points is their linear
	average with the same weight (as shown in Subsection \ref{subsec:circ_as_lin}), 
	and the normal is the normal of each of the two PNPs. Thus we conjecture that, if a modified
	scheme converges, then its smoothness equals that of the corresponding linear scheme.
\end{remark}

\subsubsection{Approximating schemes (CC, LP)}
	\begin{figure*} [!htb]
		\begin{subfigure}[b]{0.3\textwidth}
		\centering
		\includegraphics[trim={2.7cm 5.8cm 5.8cm 7cm},clip,scale=0.3]
		{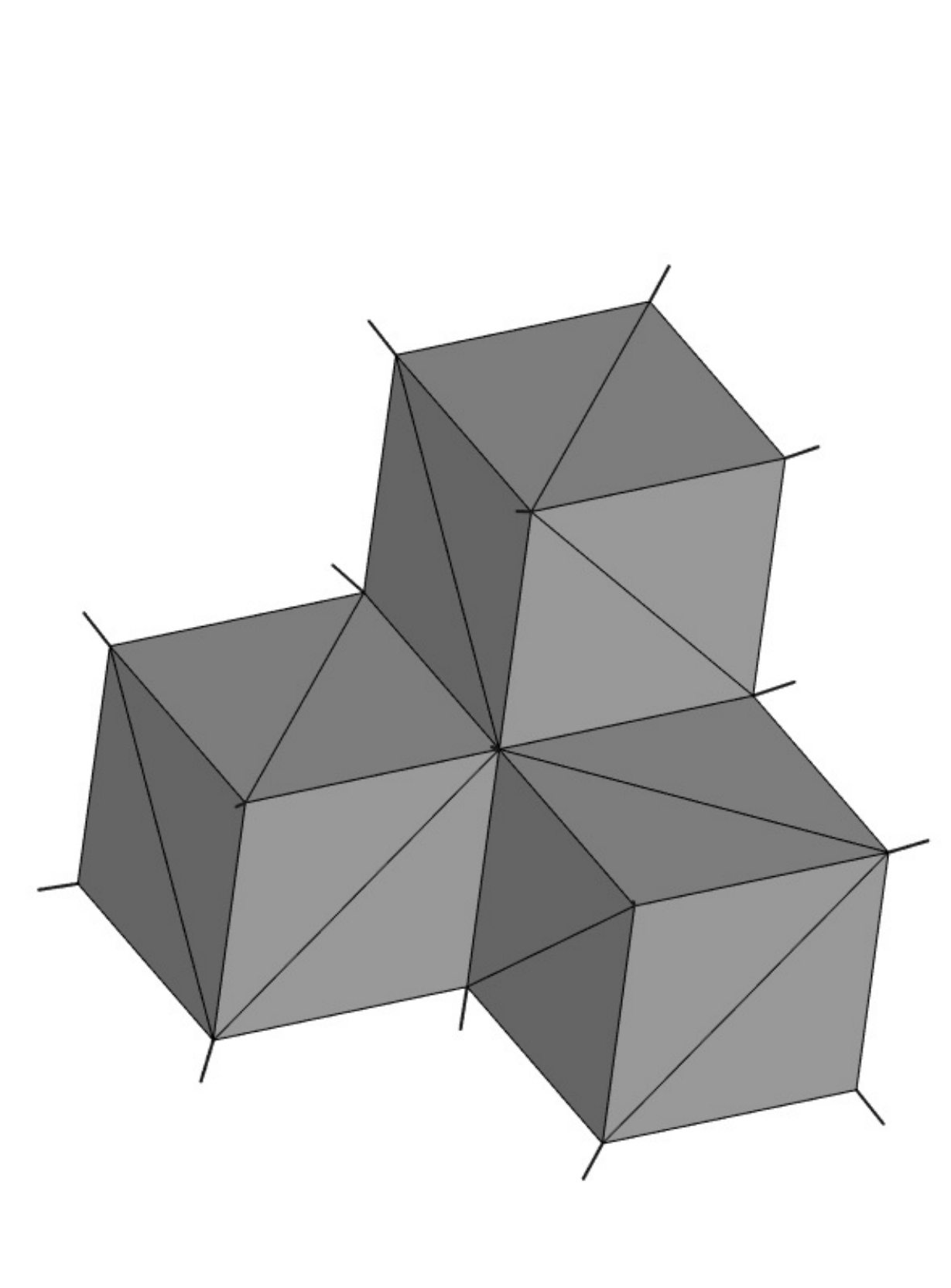}
		\caption{Input mesh with normals.}
	\end{subfigure}
	\quad
	\begin{subfigure}[b]{0.3\textwidth}
		\centering
		\includegraphics[trim={2.5cm 6.6cm 5.8cm 9cm},clip,scale=0.3]
		{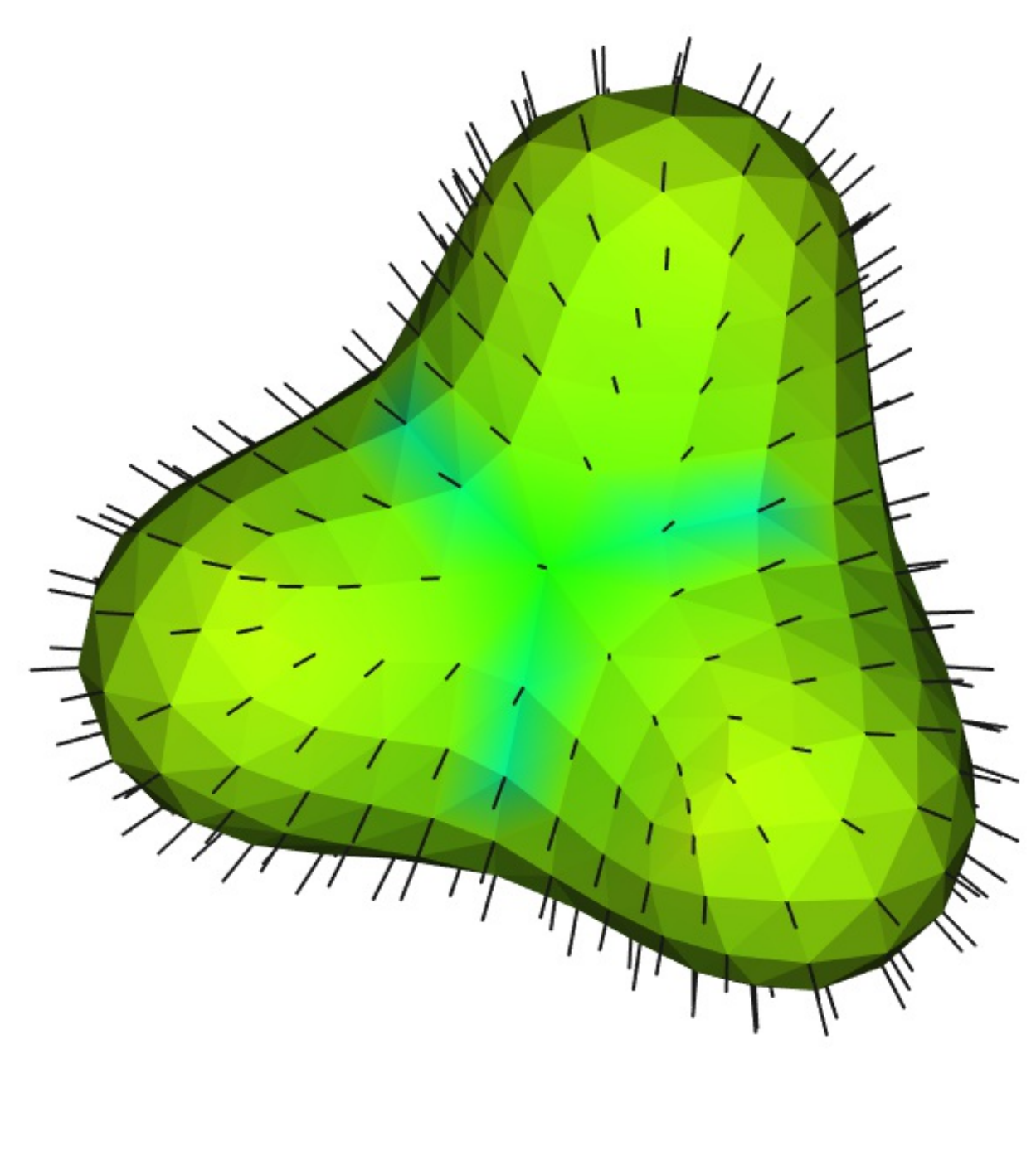}
		\caption{Modified scheme.}
	\end{subfigure}
	\quad
	\begin{subfigure}[b]{0.3\textwidth}
		\centering
		\includegraphics[trim={1.3cm 2.8cm 0cm 3.5cm},clip,scale=0.2]
		{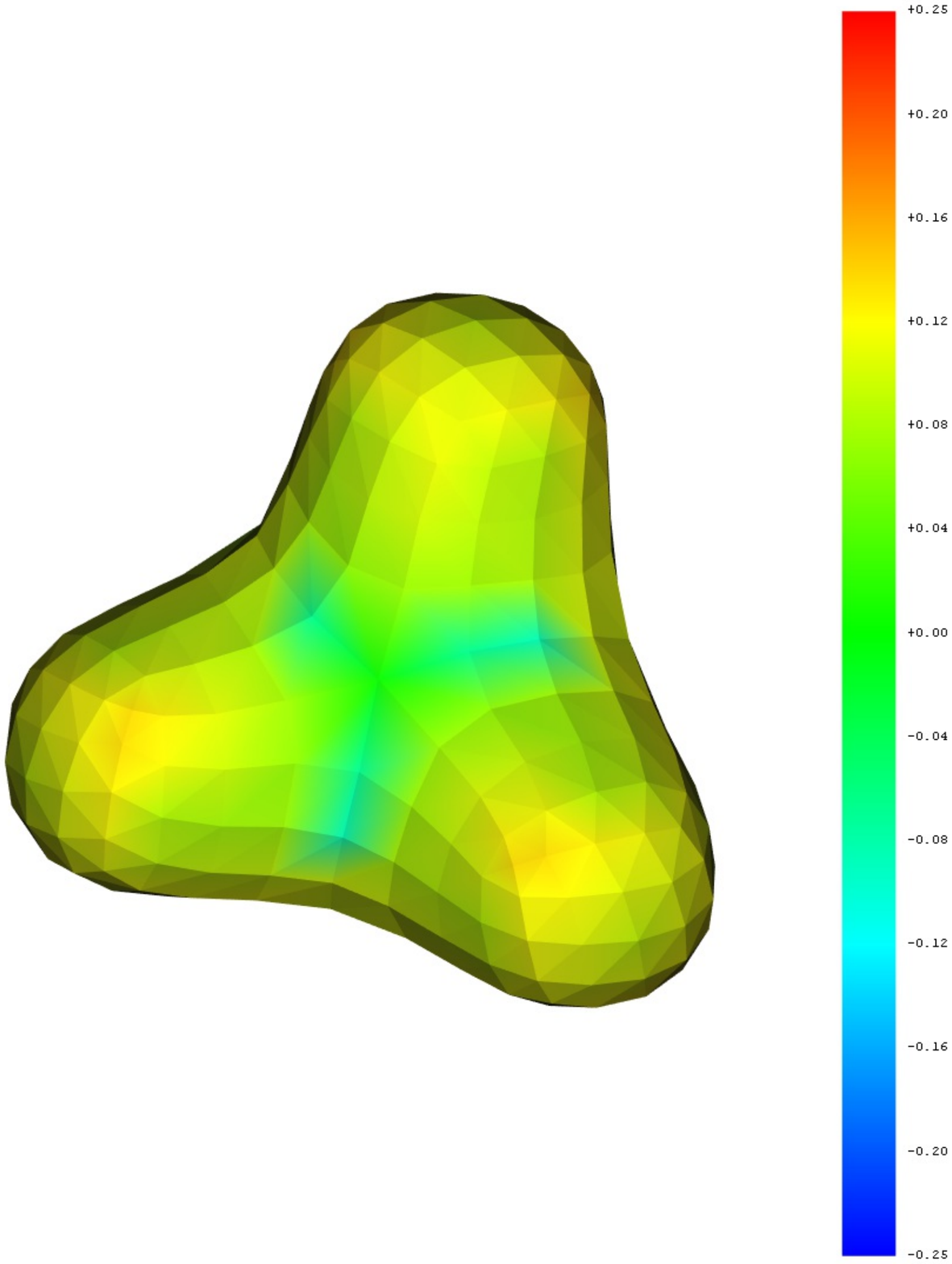}
		\caption{Linear scheme.}
	\end{subfigure}
	\captionsetup{justification=centering}
	\caption{Triangulated Tower mesh and meshes generated by MLP and LP after 2 iterations.\\
		Colors indicate magnitudes of curvature
		in the range [-0.25, 0.25].}
	\label{fig:lp_tower}
\end{figure*}
The results of the linear LP scheme and the MLP scheme, applied
to the triangulated Tower, are depicted 
in Figure \ref{fig:lp_tower}. The results of the CC and the MCC schemes,
applied to the quadrilateral Tower, look
very similar, and are barely distinguishable from the LP results. 

The meshes produced by MCC and MLP after four iterations
appear to be "blown up" versions of the meshes generated by the CC and LP
schemes respectively.
The "blown up" meshes are not contained in the convex hull of the initial
mesh, as do the linear variants.

\subsubsection{Interpolating schemes (K4, BY)}
\begin{figure*} [!htb]
	\begin{subfigure}[b]{0.3\textwidth}
		\centering
		\includegraphics[trim={2.7cm 5.8cm 6.7cm 7cm},clip,scale=0.3]
		{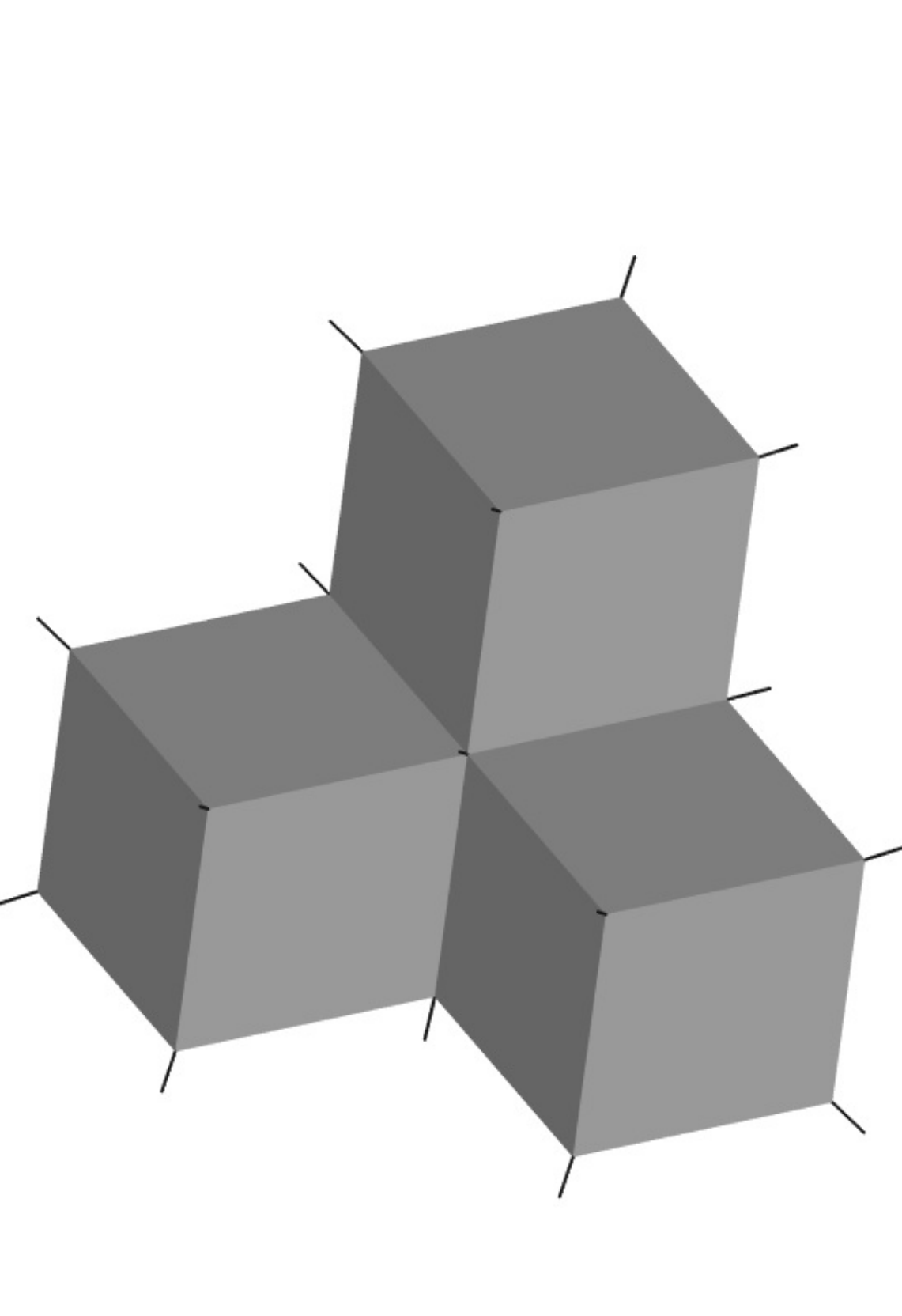}
		\caption{Input mesh with normals.}
	\end{subfigure}
	\begin{subfigure}[b]{0.3\textwidth}
		\centering
		\includegraphics[trim={2.5cm 6.6cm 5.8cm 9cm},clip,scale=0.3]
		{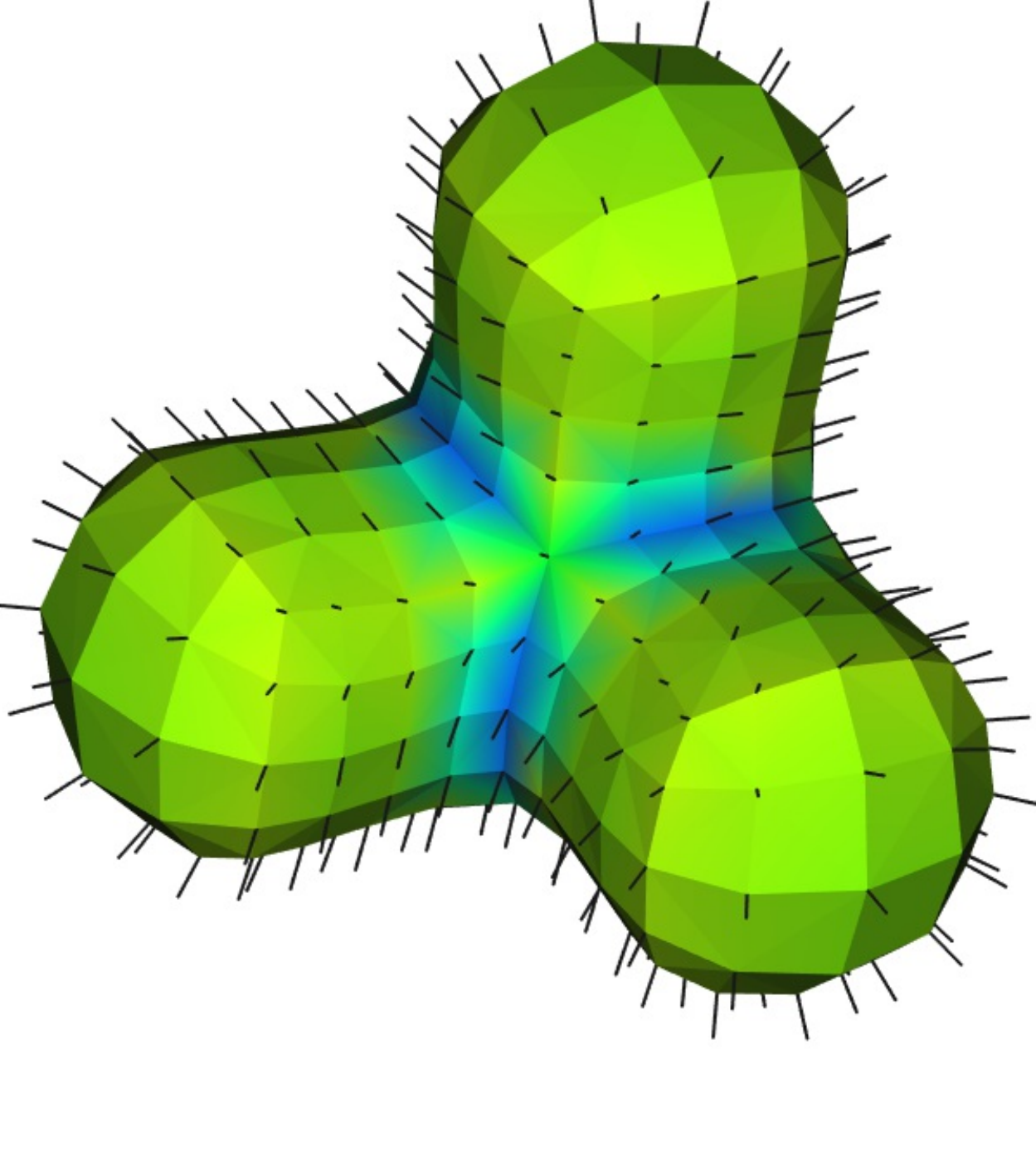}
		\caption{Modified scheme.}
	\end{subfigure}
	\quad
	\begin{subfigure}[b]{0.3\textwidth}
		\centering
		\includegraphics[trim={1.3cm 2.8cm 0cm 3.5cm},clip,scale=0.2]
		{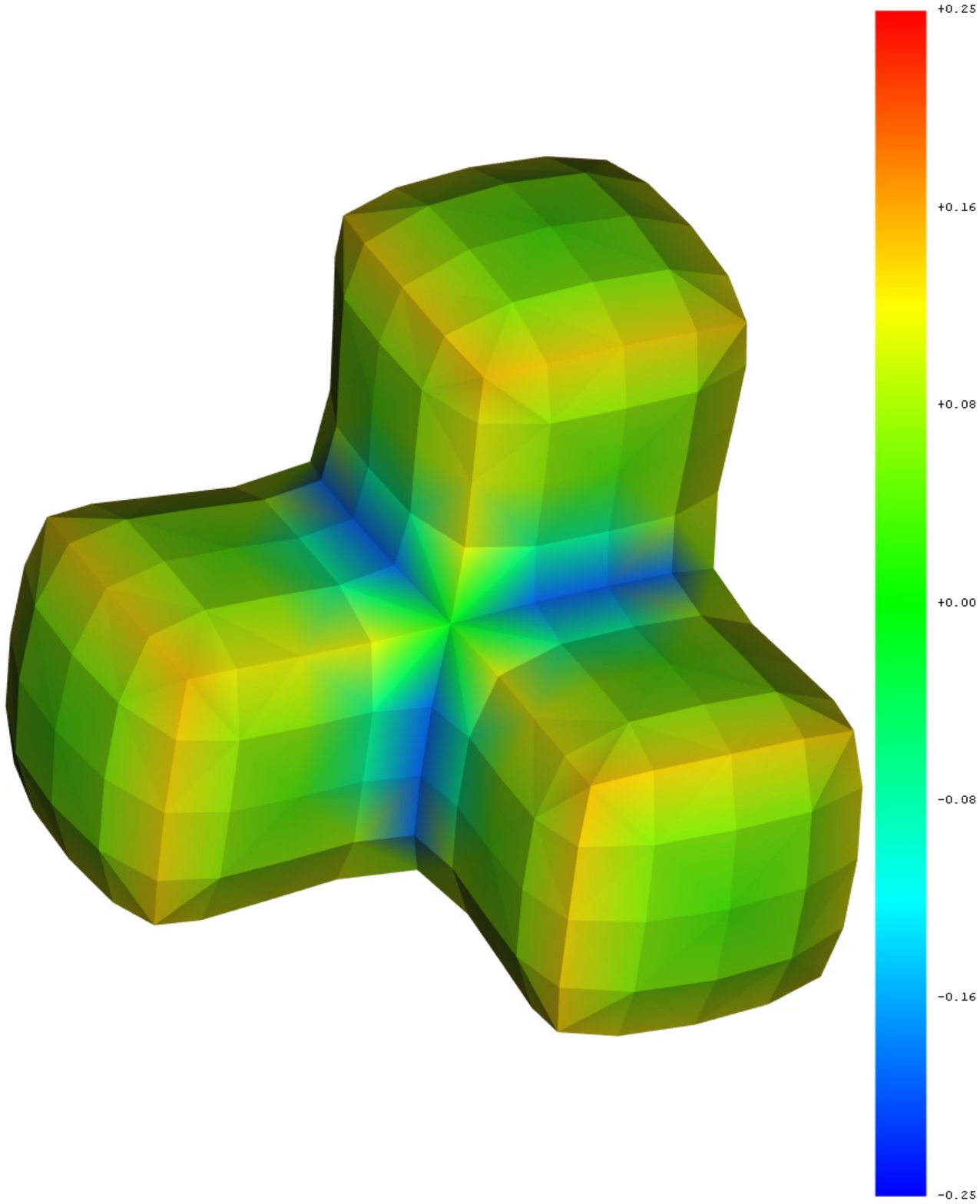}
		\caption{Linear scheme.}
	\end{subfigure}
	\captionsetup{justification=centering}
	\caption{Tower mesh and meshes generated by MK4 and K4 after 2 iterations.\\
		Colors indicate magnitudes of curvature 
		in the range [-0.25, 0.25].}
	\label{fig:kob4pt_tower}
\end{figure*}

The results in Figure \ref{fig:kob4pt_tower} were obtained by applications
of the linear and the modified K4 schemes to the tower mesh.
The MK4 scheme produces
meshes with smoother discrete curvature. However, the result
of the linear variant follows the input control mesh more accurately, as is demonstrated in Fugure \ref{fig:kob4pt_tower}.
A similar behavior is observed in the BY/MBY case, when applied to the
triangulated Tower.

\begin{figure*} 
	\begin{subfigure}[b]{0.3\textwidth}
		\centering
		\includegraphics[trim={5.5cm 10cm 2cm 9cm},clip, scale=0.35]
		{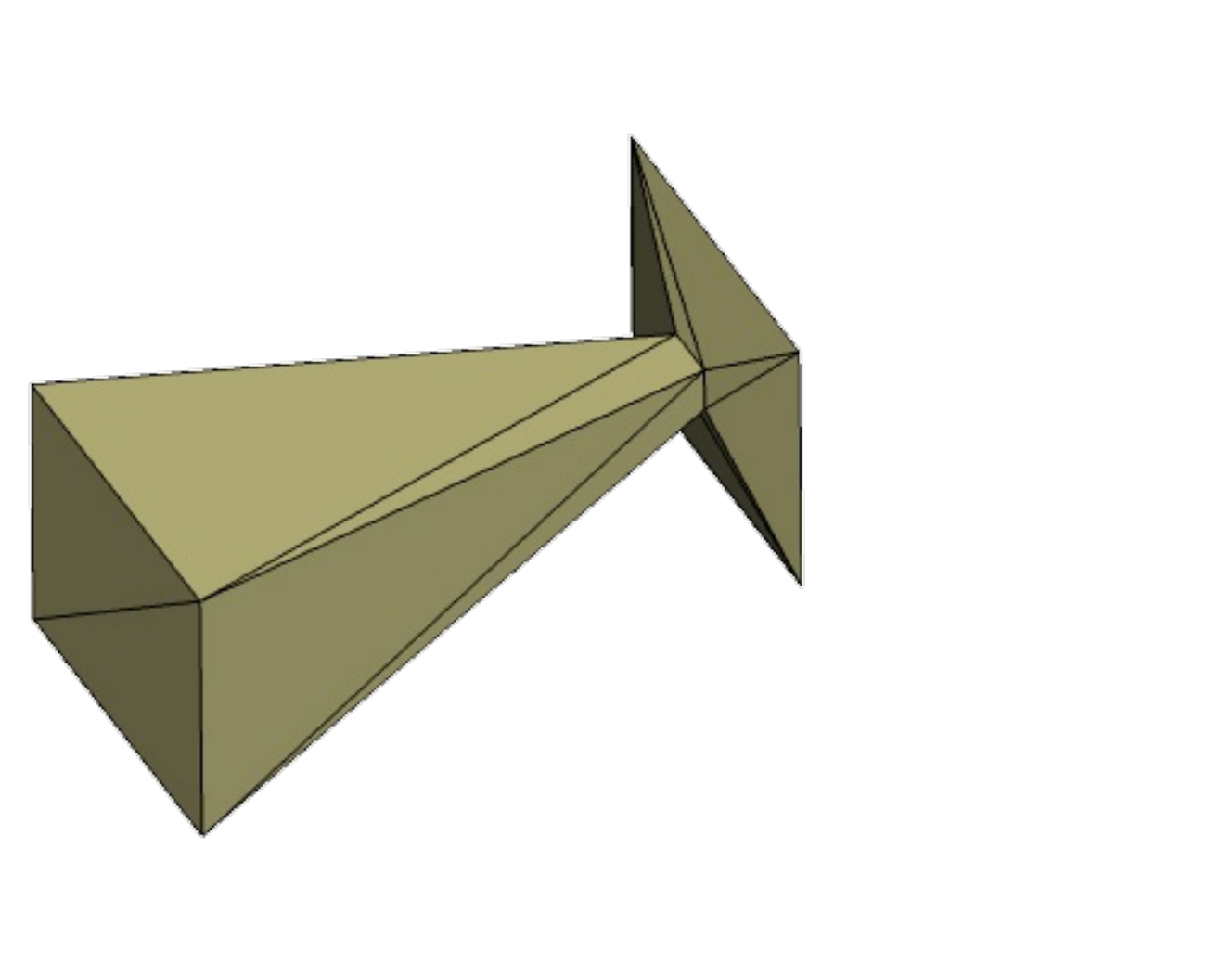} 
		\caption{Tube input mesh}
	\end{subfigure}
	\begin{subfigure}[b]{0.3\textwidth}
		\centering
		\includegraphics[trim={5cm 10cm 3cm 9cm},clip, scale=0.35]
		{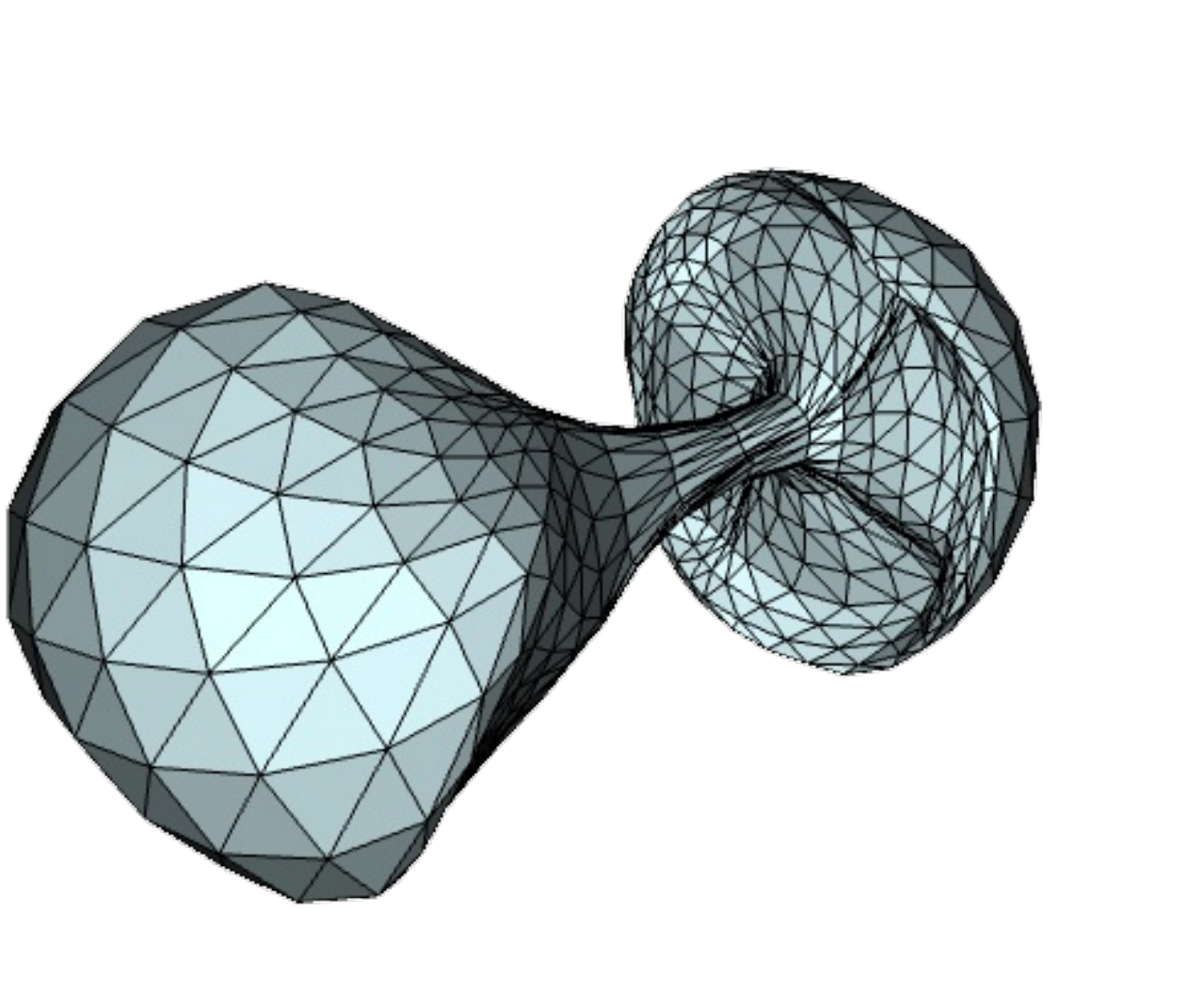} 
		\caption{Modified scheme}
	\end{subfigure}
	\quad
	\begin{subfigure}[b]{0.3\textwidth}
		\centering
		\includegraphics[trim={20cm 5cm 18cm 8cm},clip, scale=0.3]
		{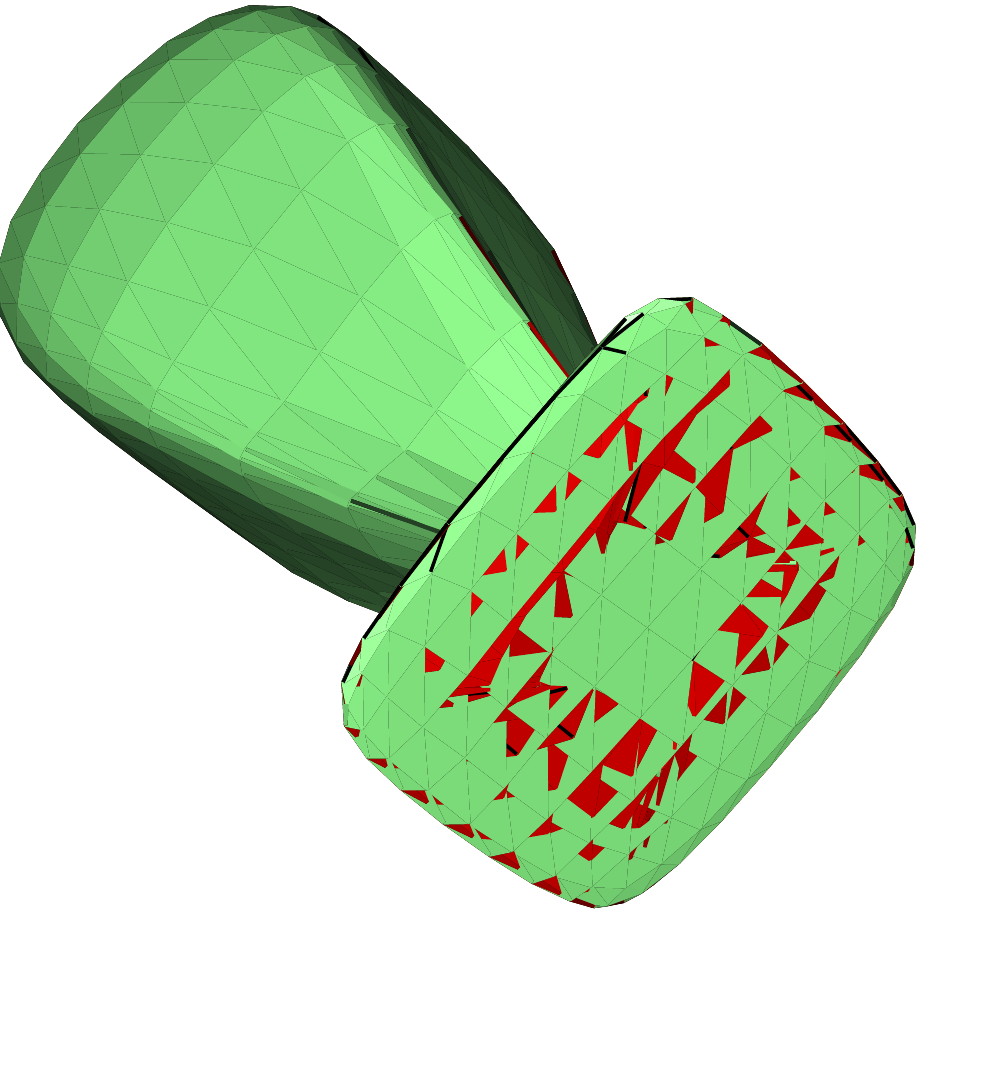} 
		\caption{Linear scheme}
	\end{subfigure}
	\captionsetup{justification=centering}
	\caption{Triangulated Tube mesh and meshes generated by MBY and BY after 3 iterations.
		\\
		Colors emphasize self intersection, or the lack of it.}
	\label{fig:butr_tube}
\end{figure*}

It is shown in \cite{dfh:09} that the linear 4-point scheme 
produces a self intersecting curve for an input polygon 
with edges of significantly different lengths. 
We observe the same 
artifact in the case of 3D meshes refined by the linear K4 and BY schemes.
This artifact is not surprising,
since both schemes are generalizations of the 4-point scheme.

An example of such performance by the BY scheme is depicted in Figure \ref{fig:butr_tube}.
The self intersections can be observed in the actual 3D  models
given in our Github repository (see Section \ref{sec:implem}).
Also in Figure \ref{fig:butr_tube}(c) this self intersection
is indicated by the red color, which is assigned to the "inner side
of the surface".
On the other hand, there are no self intersections in the surfaces
generated by the modified schemes (see Figure \ref{fig:butr_tube}(b)).

\subsection{Demonstration of the editing capabilities of the modified schemes}
\label{subsec:video}
Three videos demonstrating the variety of geometries obtainable with
modified schemes are on our Youtube channel at \cite{youtube}. The videos
show geometry morphing processes when the initial mesh is kept and
the initial normals are rotated. 
We obtain a sequence of eleven sets of initial normals,
starting from all normals equal to some normal $n^*$, 
and arriving in ten steps at the naive normals of the mesh.
The initial normal at a mesh point in case $i, i = 0,...,10,$ is the weighted 
geodesic average between $n^*$ and the naive normal
at that point, with weight $\mu_i = i/10$.
Refining four times each set of initial data by the same modified subdivision scheme,
a sequence of geometries is obtained. These geometries are combined 
into a morphing video, demonstrating changes from a 
surface generated by the corresponding linear scheme (in the case when 
all initial normals are equal to $n^*$, see Section 
\ref{subsec:circ_as_lin}), to the surface generated from the naive normals.
Three snapshots of one of the videos are given in Figure \ref{fig:lp_edit}.
Note the changes in the curvatures due to the changes in the initial normals.

We compute a numerical quantity which estimates the "distance" between
the limit normals and the normals of the limit surface, for the different 
initial normals in cases $i=0,1,...,10$.
For every PNP in a final mesh, we compute the angle between
the calculated normal of the PNP and the naive normal in the final mesh
at the point of the PNP (approximating the normal of the limit surface there).

Table \ref{tbl:angle_err} contains averages of these angles for the
weights $\mu_i$, denoted by $\xi_i, i = 0,...,10$. Note, 
that $\xi_i$ decrease monotonically with $i$.
\begin{table}[h]
	\begin{center}
\begin{tabular} { | c|c | c | c | c | c | c | c | c |c | c | c| }
	\hline 
	& & & & & & & & & & &\\
	$\mu_i$ &
	0.0 &
	0.1 & 
	0.2 & 
	0.3 & 
	0.4 & 
	0.5 & 
	0.6 & 
	0.7 & 
	0.8 & 
	0.9 & 
	1.0 \\
	& & & & & & & & & & &\\
	\hline
	& & & & & & & & & & &\\
	$\xi_i$&
85\textdegree&
77\textdegree&
71\textdegree&
64\textdegree&
56\textdegree&
48\textdegree&
38\textdegree&
29\textdegree&
22\textdegree&
16\textdegree&
12\textdegree\\
	& & & & & & & & & & &\\
	\hline
\end{tabular}
\end{center}
\captionsetup{justification=centering}
\caption{Averages of the angles between the generated normals at the final \\
	refinement level
	and the corresponding naive normals of the final mesh.}
\label{tbl:angle_err}
\end{table}
\begin{figure*}
	\begin{subfigure}[b]{0.3\textwidth}
		\centering
		\includegraphics[trim={4cm 10cm 2.5cm 8.5cm},clip, scale=0.3]
		{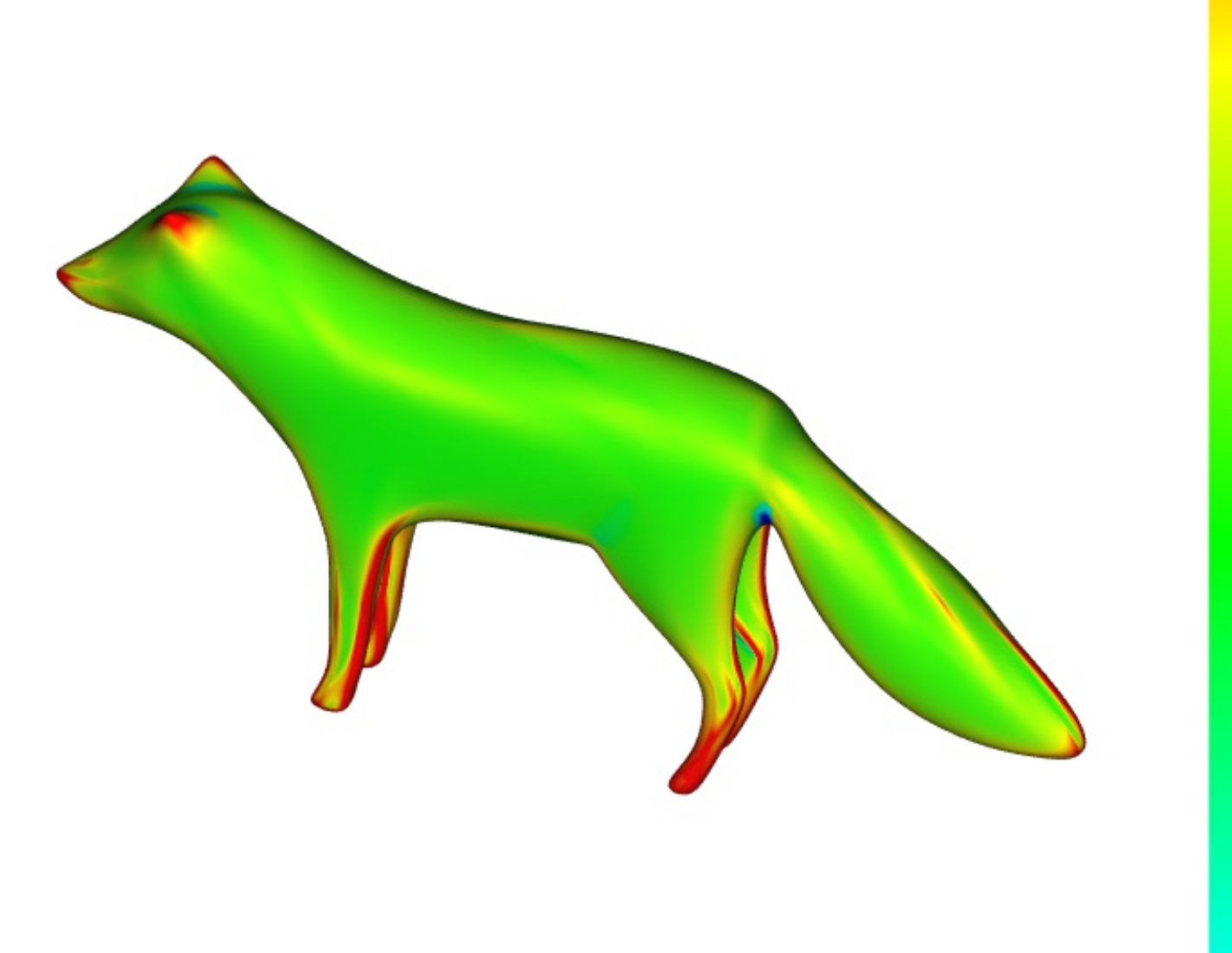} 
		\caption{}
	\end{subfigure}
	\quad
	\begin{subfigure}[b]{0.3\textwidth}
		\centering
		\includegraphics[trim={4cm 10cm 3.5cm 8.5cm},clip, scale=0.3]
		{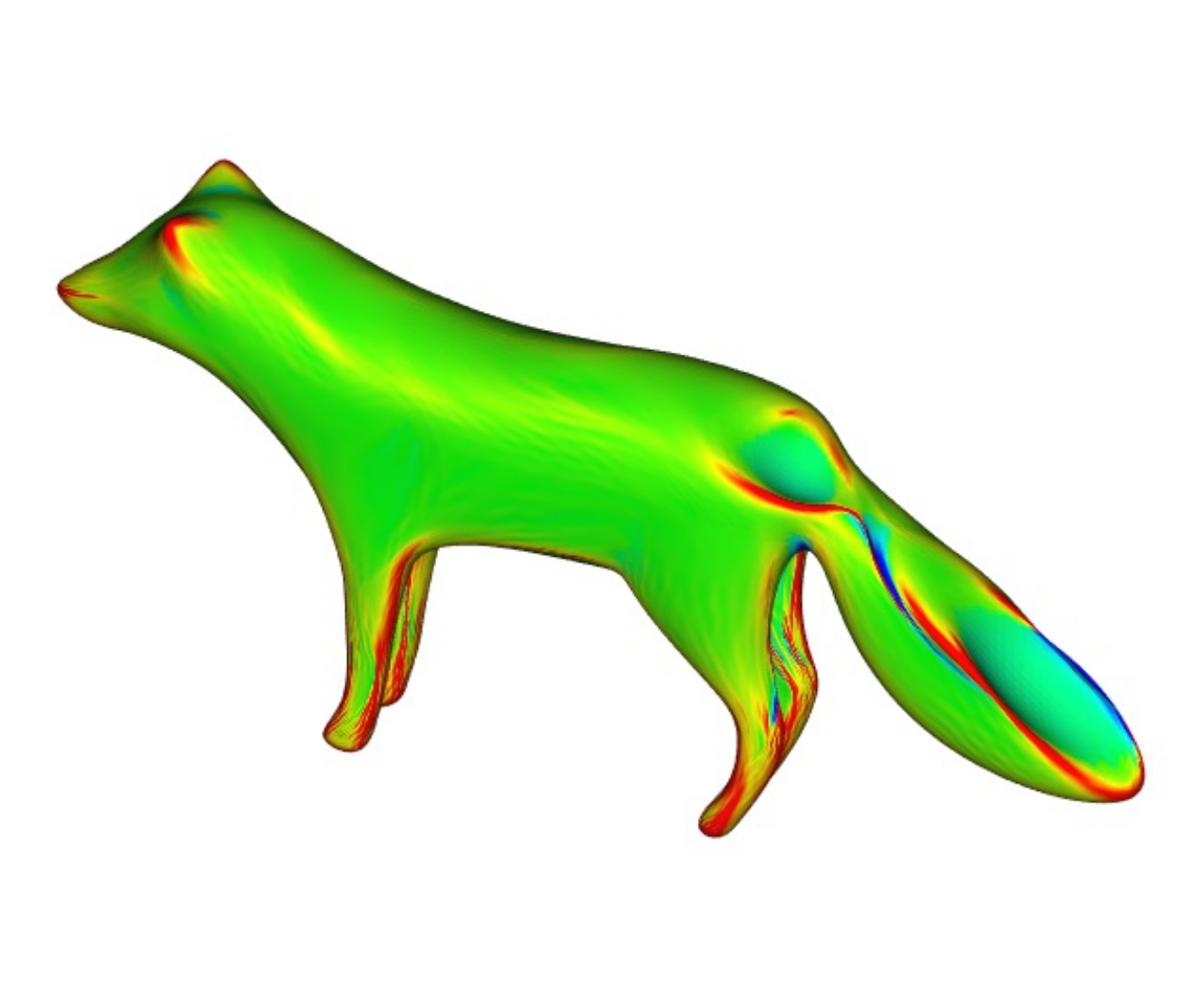} 
		\caption{}
	\end{subfigure}
	\quad
	\begin{subfigure}[b]{0.3\textwidth}
		\centering
		\includegraphics[trim={4cm 10cm 3.5cm 8.5cm},clip, scale=0.3]
		{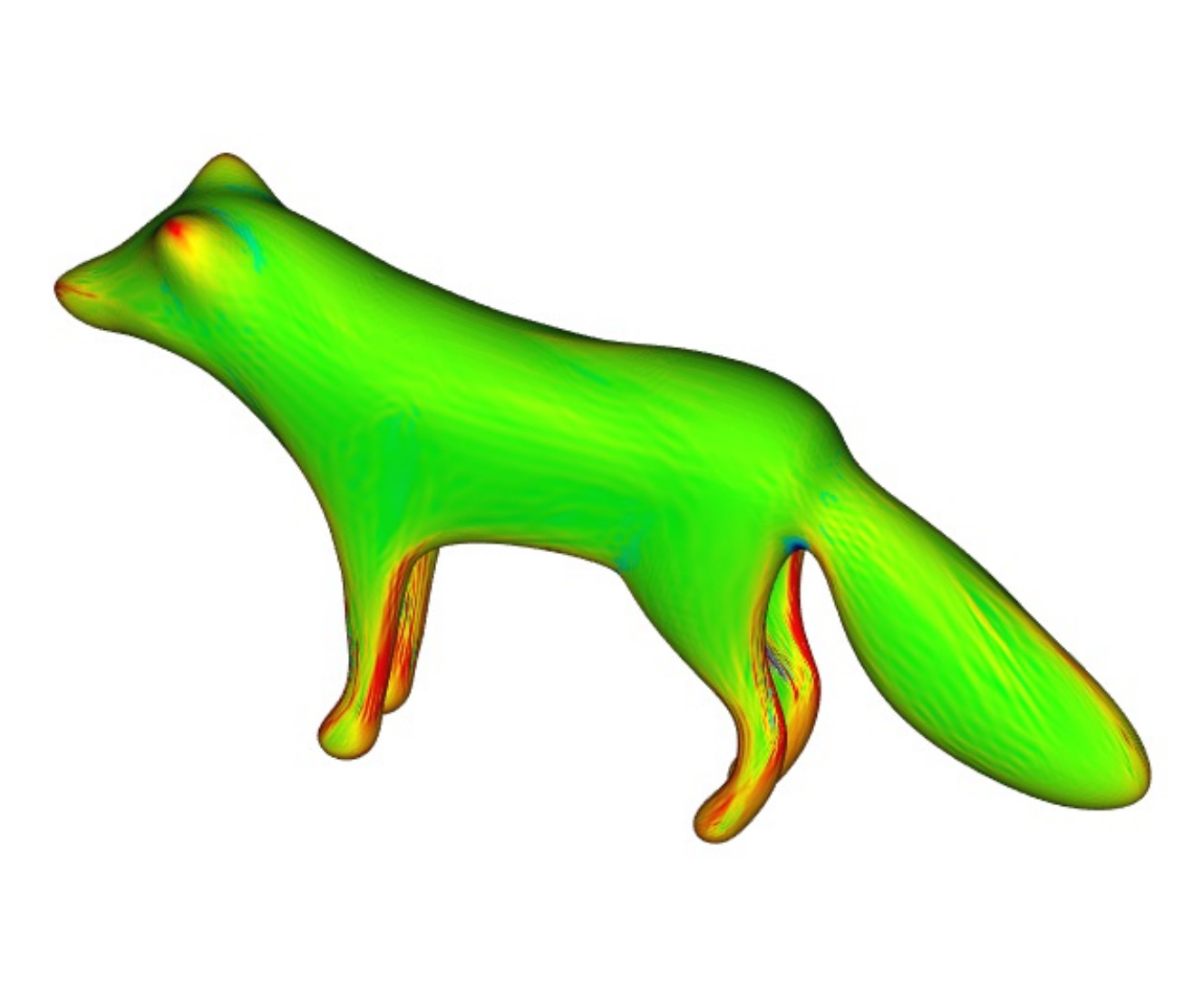} 
		\caption{}
	\end{subfigure}
	\captionsetup{justification=centering}
	\caption{ Surfaces generated by MLP, starting from the same "fox" mesh 
		with different initial normals.
		Colors indicate magnitudes in the range  [-0.5, 0.5] of curvature at the vertices of the final refined mesh, generated by four iterations.}
	\label{fig:lp_edit}
\end{figure*}

Table \ref{tbl:angle_err} indicates that the closer are the initial normals 
to the naive normals,
the closer are the limit normals to the normals of the limit surface. 

Another observation is that one can setup initial normals such that a modified
scheme computes a surface with unexpected geometry. See Figure \ref{fig:lp_edit}b,
for an example. Observe that in Figure \ref{fig:lp_edit} the main changes are
in the tail of the fox.

\section{Implementation}
\label{sec:implem}
We provide an implementation of the algorithms and  
comparison methods, developed in this work, in our Github repository at \cite{github}.
The implementation is in the Python language. The input and the refined mesh files
of the examples studied in this paper are in that repository too.  

\section{Conclusions and future work}
\label{sec:future}
In this paper we design an extension  of the 2D circle average to 3D. This is 
indeed an extension, since the 3D circle average
coincides with the 2D circle average when the two averaged PNPs are in the
same plane.
We modify surface-generating linear subdivision schemes refining points 
to surface-generating schemes refining PNPs, using the 3D circle average.
The modified schemes can generate a variety of new geometries from a given 
mesh, by editing the initial normals. These editing capabilities are 
demonstrated in Figure \ref{fig:lp_edit} and by three videos \cite{youtube}.  

Our investigation of the performance of the modified schemes included
more examples than those of Section \ref{subsec:results}.
An overall observation is that the results of a modified scheme with naive 
normals appear to be smoother than those of the 
corresponding 
linear scheme 
in case the initial mesh is homogeneous (consists of edges with lengths
of the same order of magnitude).
Also, the results of the modified approximating schemes are 
"blown up" versions of their 
corresponding 
linear counterparts, and are
not necessarily 
contained
in the convex hull of the initial mesh. 
\\

\noindent Several research directions should be addressed  in the future:
\begin{itemize}
	\item 
	To prove the convergence of the modified schemes.
	\item 
	To analyze the smoothness of the modified schemes.
\end{itemize}
These two topics are addressed in \cite{ld:16},\cite{ld:19} for certain 2D schemes.
\begin{itemize}
	\item  
	How to support creases/sharp edges with the modified schemes?
	\item 
	To design a new binary operation between two point-normal pairs 
	such that the modified schemes with this operation  generate limit normals which are
	the normals of the limit surface,
	a property not possessed by our modified schemes (see Section \ref{subsec:results}).	
\end{itemize}

\end{document}